\newif\ifiopclass
\IfFileExists{iopart.cls}{\iopclasstrue\documentclass[11pt]{iopart}}{\iopclassfalse\documentclass[11pt]{article}}

\usepackage{amssymb,amsfonts,amsmath}
\usepackage{tikz}
\usepackage{enumitem,hyperref}
\usepackage[numbers,sort&compress]{natbib}
\usetikzlibrary{arrows.meta,calc}
\usepackage{amsthm}
\IfFileExists{iopams.sty}{\usepackage{iopams}}{}

\newtheorem{assumption}{Assumption}
\newtheorem{definition}{Definition}

\newtheorem{proposition}{Proposition}
\newtheorem{lemma}{Lemma}
\theoremstyle{remark}         
\newtheorem{remark}{Remark}   

\begin{document}

\title{Brachistochrone-ruled timelike surfaces in Newtonian and relativistic spacetimes}

\ifiopclass
\author{Ferhat Taş}
\address{İstanbul University, Department of Computer Science, İstanbul 34134, Türkiye}
\ead{tasf@istanbul.edu.tr}
\else
\author{Ferhat Taş\\
\normalsize İstanbul University, Department of Computer Science, İstanbul 34134, Türkiye\\
\texttt{tasf@istanbul.edu.tr}}
\date{}
\maketitle
\fi

\begin{abstract}
We present brachistochrone-ruled timelike surfaces as a natural geometric framework combining ruled-surface geometry with arrival-time minimization: two-dimensional surfaces whose rulings are arrival-time-minimizing timelike curves between paired endpoint families. The paper has two goals: to formalize this framework and to show that it is workable in both flat and curved spacetimes. We begin with a Newtonian toy model built from cycloidal brachistochrones, then pass to stationary Lorentzian spacetimes where the arrival-time problem reduces to a spatial first-order variational problem. In general the reduced integrand is a smooth Lagrangian, while static fixed-energy cases recover Jacobi-type Riemannian metrics. As a consistency check, the bounded-speed Minkowski model yields straight timelike rulings and a planar totally geodesic example. In the Schwarzschild exterior, fixed-energy coordinate-time minimization is reduced to geodesics of an explicit Jacobi metric, and we provide a practical numerical workflow for constructing the associated ruled surfaces. Finally, by linearizing the reduced Euler--Lagrange equations along rulings, we obtain the natural transverse variation equation and boundary data, linking the framework to Jacobi-field geometry and stability questions.
\end{abstract}

\providecommand{\keywords}[1]{\par\noindent\textit{Keywords: }#1\par}
\keywords{brachistochrone, ruled surface, stationary spacetime, variational reduction, Jacobi metric}

\ifiopclass\maketitle\fi

\section{Introduction}
\label{sec:intro}
Geodesics play a central role in both classical differential geometry and general relativity. In a Riemannian manifold they locally minimize length; in a Lorentzian spacetime, timelike geodesics represent the worldlines of freely falling observers and locally extremize the proper time between events. From the variational viewpoint, geodesics are critical points of natural action functionals built from the metric tensor.

From a physical perspective, brachistochrone-ruled timelike surfaces may be interpreted as \emph{time-optimal world sheets} connecting two families of sources and receivers, or two congruences of observers. They encode, in a single geometric object, the set of time-minimizing trajectories which realize the fastest possible transport or signal propagation between corresponding elements of the two families. Such interpretations have direct applications in gravitational physics, for instance in modeling optimal signal paths in curved spacetimes (e.g., near black holes) or analyzing congruences of observers in contexts like gravitational wave detection and wavefront propagation. From a geometric viewpoint, such surfaces interpolate between several known constructions: in flat spacetime with trivial time functional they reduce to timelike surfaces ruled by straight geodesics, while in stationary spacetimes they are related to geodesics of reduced spatial variational structures, which in special regimes are described by Randers/Finsler or Jacobi metrics \citep{BaoChernShen00,Randers41,Zermelo31,BaoRoblesShen04,GibbonsHerdeiroWarnickWerner09,Gibbons16,CaponioJavaloyesSanchez24}.

In many applications, the relevant variational quantity is not length or proper time but \emph{arrival time} measured by a preferred observer family. This leads to brachistochrone-type problems: for a given time functional and two endpoints, find the timelike trajectory with minimal travel time. In the Newtonian case, this starts with the classical cycloidal brachistochrone \citep{Bernoulli1696,Erlichson99}. In stationary relativistic settings, Perlick formulated the corresponding arrival-time principle \citep{Perlick90,Perlick91}, while existence and variational analyses were developed by Giannoni--Piccione and related works \citep{FortunatoGiannoniMasiello95,GiannoniPiccione98,CaponioJavaloyesMasiello11,CaponioJavaloyesSanchez11}. Relativistic refinements of brachistochrone dynamics also appear in \cite{GoldsteinBender86}.

Independently, ruled-surface geometry has a long history in both Riemannian and Lorentzian contexts, including surfaces ruled by geodesics and null geodesics. These surfaces play important roles in geometry and physics (e.g., wavefronts, causal structures, and string-like models).

Motivated by this background, we seek a systematic framework in which each ruling of a smooth two-parameter timelike surface is a time-minimizing trajectory for a moving endpoint pair. This viewpoint is particularly useful in stationary spacetimes, where arrival-time minimization admits reduced spatial Lagrangian or Jacobi descriptions.

The aim of this paper is to combine these two variational ideas into a systematic construction of \emph{brachistochrone-ruled timelike surfaces} in classical and relativistic spacetimes. Roughly speaking, given two one-parameter families of events or observers, we consider for each parameter value the time-minimizing timelike curve (brachistochrone) connecting the corresponding pair of endpoints, with respect to a chosen time functional. The union of these time-minimizing worldlines forms a two-dimensional timelike surface in spacetime, ruled by brachistochronal rulings. Our goal is to formalize this construction, to relate it to known reductions of arrival-time variational problems to spatial Lagrangian formulations (and to Jacobi/Finsler metrics in special cases), and to explore its geometry in simple but illustrative examples.

\medskip
The main contributions of this work can be summarized as follows:
\begin{itemize}
  \item We introduce a precise definition of
        \emph{relativistic brachistochrone-ruled timelike surfaces}
        in stationary Lorentzian spacetimes. Using reductions of
        arrival-time functionals to spatial variational problems, we
        characterize the rulings as extremals of the reduced spatial
        functional and formulate natural regularity and uniqueness
        assumptions for families of brachistochrones. In static
        fixed-energy settings this specializes to Jacobi geodesics.
  \item To build intuition, we first construct a fully explicit
        Newtonian toy model. In a classical constant-gravity field we
        use cycloidal brachistochrones as rulings and obtain a
        ``brachistochrone-ruled surface'' in Newtonian spacetime.
        This provides a concrete variational picture in which the
        rulings are honest time-minimizing curves with respect to the
        classical brachistochrone functional.
  \item In $(1+2)$-dimensional Minkowski spacetime we consider a simple
        arrival-time functional based on a uniform bound on the
        spatial speed relative to static observers. We show that the
        corresponding relativistic brachistochrones are straight
        timelike lines and construct explicit brachistochrone-ruled
        timelike surfaces connecting two families of stationary
        observers. A particularly simple choice of boundary curves
        yields a timelike affine plane which is totally geodesic.
  \item As a first non-trivial curved example, we study the Schwarzschild
        exterior as a static spacetime. For timelike geodesics at fixed
        energy we derive a Jacobi-type Riemannian metric on the
        equatorial spatial slice and explain how time-minimizing
        timelike geodesics (with respect to coordinate time) reduce to
        length-minimizing geodesics of this Jacobi metric. We then
        outline a numerical scheme for constructing brachistochrone-ruled
        timelike surfaces between two families of endpoints by solving
        boundary-value problems for the Jacobi geodesics.
  \item Finally, we discuss the differential geometry of
        brachistochrone-ruled timelike surfaces, emphasizing the role
        of the induced Lorentzian metric, the second fundamental form
        and curvature invariants. We comment on the relation to
        conjugate points, cut loci and caustics of the underlying
        time-minimizing geodesics, and we point out several directions
        for further work.
\end{itemize}
\medskip
The paper is organized as follows. Section~\ref{sec:materials_methods}
develops the Newtonian toy model and gives an explicit
brachistochrone-ruled surface built from cycloids. In
Section~\ref{sec:relativistic_brachistochrone_ruled} we formulate the
stationary relativistic framework and define brachistochrone-ruled
timelike surfaces via reduced spatial variational problems.
Section~\ref{sec:minkowski_example} provides a flat-spacetime
consistency check, while Section~\ref{sec:schwarzschild_example}
contains the Schwarzschild fixed-energy Jacobi reduction together with
a practical numerical construction pipeline. In
Section~\ref{sec:geometry_brachistochrone_ruled} we give a more
quantitative geometric analysis (induced metric, second fundamental
form, curvature identities, and Jacobi-field boundary data). We
conclude in Section~\ref{sec:conclusion} with open problems and future
directions.

\section{Newtonian toy model: cycloidal brachistochrone rulings}
\label{sec:materials_methods}

\begin{figure}
    \centering
    \includegraphics[width=0.8\linewidth]{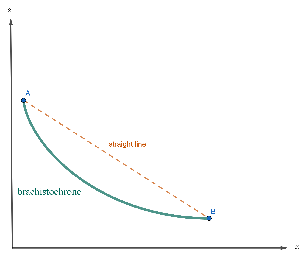}
    \caption{Standard brachistochrone curve connecting two points in a uniform gravitational field.}
    \label{fig:newton_brachistochrone}
\end{figure}

This section presents a Newtonian toy model that captures the core ruled-surface construction and serves as a template for the relativistic setting.

\subsection{Classical brachistochrones revisited}

We recall the classical planar brachistochrone.
Let $(x,z)$ be Cartesian coordinates, with $z$ measured downward in a constant gravitational field of strength $g>0$. A frictionless particle moves under the potential
\begin{equation}
  V(z) = g z,
\end{equation}
starting from rest at
\begin{equation}
  A = (0,0)
\end{equation}
and reaching
\begin{equation}
  B = (X,H), \qquad X>0,\ H>0.
\end{equation}

Let $\gamma:[0,1]\to\mathbb{R}^2$, $\gamma(\lambda)=(x(\lambda),z(\lambda))$
be a smooth curve connecting $A$ to $B$.
If the particle has no initial velocity, conservation of mechanical energy gives
\begin{equation}
  \frac{1}{2} v^2(\lambda) = g z(\lambda), \qquad
  v(\lambda) = \sqrt{2 g\,z(\lambda)},
\end{equation}
where $v=\|\dot\gamma\|$ is the speed. The travel time along $\gamma$ is
\begin{equation}
  T[\gamma]
  \;=\; \int_0^1 \frac{\mathrm{d}s}{v}
  \;=\; \int_0^1 \frac{\sqrt{\dot x(\lambda)^2 + \dot z(\lambda)^2}}
                         {\sqrt{2 g\, z(\lambda)}}\,\mathrm{d}\lambda.
  \label{eq:brachistochrone_functional}
\end{equation}

The minimizers of~\eqref{eq:brachistochrone_functional} are cycloids: there exist parameters $a>0$ and $\theta_1>0$ such that
\begin{equation}
  \begin{aligned}
    x(\theta) &= a\bigl(\theta - \sin\theta\bigr),\\
    z(\theta) &= a\bigl(1 - \cos\theta\bigr),
  \end{aligned}
  \qquad 0 \leq \theta \leq \theta_1,
  \label{eq:cycloid_param}
\end{equation}
with the endpoint conditions
\begin{equation}
  a\bigl(\theta_1 - \sin\theta_1\bigr) = X,
  \qquad
  a\bigl(1 - \cos\theta_1\bigr) = H.
  \label{eq:cycloid_boundary_conditions}
\end{equation}

Along the cycloid~\eqref{eq:cycloid_param} the speed is
\begin{equation}
  v(\theta) = \sqrt{2 g\, z(\theta)} 
  = \sqrt{2 g a\bigl(1 - \cos\theta\bigr)}.
\end{equation}
Direct computation gives
\begin{equation}
  \mathrm{d}s
  \;=\; \sqrt{\bigl(x'(\theta)\bigr)^2 + \bigl(z'(\theta)\bigr)^2}\,\mathrm{d}\theta
  \;=\; 2 a \sin\frac{\theta}{2}\,\mathrm{d}\theta,
\end{equation}
\begin{equation}
  \mathrm{d}t = \frac{\mathrm{d}s}{v}
  = \frac{2 a \sin(\theta/2)\,\mathrm{d}\theta}
         {\sqrt{2 g a(1-\cos\theta)}}
  = \sqrt{\frac{a}{g}}\,\mathrm{d}\theta.
\end{equation}
Hence the time coordinate is linear in $\theta$:
\begin{equation}
  t(\theta) = \sqrt{\frac{a}{g}}\;\theta, \qquad t(0)=0.
  \label{eq:time_param_cycloid}
\end{equation}

In Newtonian spacetime $\mathbb{R}^4$ with coordinates $(t,x,z,y)$, this gives the worldline
\begin{equation}
  \gamma(\theta)
  = \bigl( t(\theta), x(\theta), z(\theta), y_0 \bigr)
  = \Bigl(
        \sqrt{\tfrac{a}{g}}\;\theta,\;
        a(\theta - \sin\theta),\;
        a(1 - \cos\theta),\;
        y_0
    \Bigr),
  \label{eq:single_worldline}
\end{equation}
where the additional coordinate $y_0\in\mathbb{R}$ is constant (Figure~\ref{fig:newton_brachistochrone}).

\subsection{From a family of endpoint pairs to a brachistochrone-ruled surface}

We now use these classical solutions to build a two-dimensional surface ruled by brachistochrone worldlines.

Let $I\subset\mathbb{R}$ be a non-empty open interval and let
\begin{equation}
  X, H : I \longrightarrow (0,\infty)
\end{equation}
be smooth functions. We define two spatial curves in $\mathbb{R}^3$ by
\begin{equation}
  \Gamma_0 : I \to \mathbb{R}^3,\quad
  \Gamma_0(s) = (x,z,y) = \bigl(0,0,s\bigr),
  \label{eq:gamma0}
\end{equation}
\begin{equation}
  \Gamma_1 : I \to \mathbb{R}^3,\quad
  \Gamma_1(s) = (x,z,y) = \bigl(X(s), H(s), s\bigr).
  \label{eq:gamma1}
\end{equation}
For each $s\in I$, the points
\begin{equation}
  A_s = (0,0),\qquad B_s = \bigl(X(s), H(s)\bigr)
\end{equation}
define a brachistochrone endpoint pair in the $(x,z)$-plane.

\begin{assumption}
\label{ass:brachistochrone_existence_uniqueness}
For every $s\in I$ there exists a unique cycloidal brachistochrone
connecting $A_s$ to $B_s$, determined by parameters
$a(s)>0$ and $\theta_1(s)>0$ such that
\begin{equation}
  a(s)\bigl(\theta_1(s) - \sin\theta_1(s)\bigr) = X(s),
  \qquad
  a(s)\bigl(1 - \cos\theta_1(s)\bigr) = H(s).
  \label{eq:a_theta1_equations}
\end{equation}
Moreover, the functions $a,\theta_1 : I\to(0,\infty)$ are smooth.
\end{assumption}

Under Assumption~\ref{ass:brachistochrone_existence_uniqueness}, the corresponding time-minimizing trajectory in Newtonian spacetime is
\begin{equation}
  \gamma_s : [0,\theta_1(s)] \to \mathbb{R}^4,
\end{equation}
where
\begin{align}
  \gamma_s(\theta)
  &= \bigl(t_s(\theta), x_s(\theta), z_s(\theta), y_s(\theta)\bigr) \label{eq:gamma_s_theta} \\
  &= \Bigl(
        \sqrt{\tfrac{a(s)}{g}}\;\theta,\;
        a(s)\bigl(\theta - \sin\theta\bigr),\;
        a(s)\bigl(1 - \cos\theta\bigr),\;
        s
     \Bigr).
\end{align}

for $0\leq \theta \leq \theta_1(s)$.

Introduce the normalized parameter $u\in[0,1]$ by
\begin{equation}
  \theta(s,u) := u\,\theta_1(s).
\end{equation}
Define the two-parameter map
\begin{equation}
  \Sigma : I\times[0,1] \longrightarrow \mathbb{R}^4,
  \qquad
  \Sigma(s,u) := \gamma_s\bigl(\theta(s,u)\bigr).
  \label{eq:Sigma_definition}
\end{equation}
Explicitly,
\begin{equation}
  \Sigma(s,u)
  = \bigl( t(s,u), x(s,u), z(s,u), y(s,u) \bigr),
\end{equation}
with
\begin{equation}
  \begin{aligned}
    \theta(s,u) &= u\,\theta_1(s),\\[2pt]
    x(s,u) &= a(s)\bigl(\theta(s,u) - \sin\theta(s,u)\bigr),\\[2pt]
    z(s,u) &= a(s)\bigl(1 - \cos\theta(s,u)\bigr),\\[2pt]
    t(s,u) &= \sqrt{\tfrac{a(s)}{g}}\;\theta(s,u)
            = \sqrt{\tfrac{a(s)}{g}}\;u\,\theta_1(s),\\[2pt]
    y(s,u) &= s.
  \end{aligned}
  \label{eq:Sigma_components}
\end{equation}

\begin{figure}
    \centering
    \includegraphics[width=0.9\linewidth]{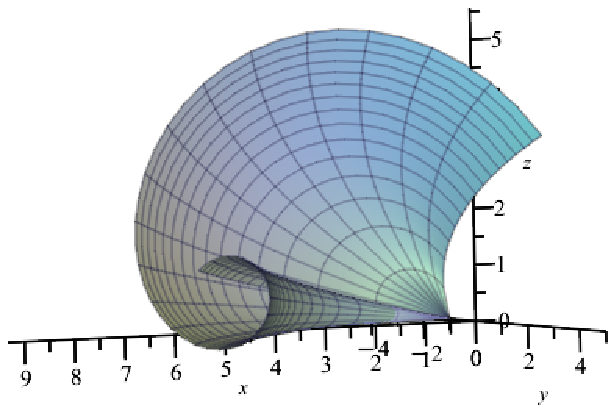}
    \caption{Brachistochrone-ruled surface in Newtonian spacetime.}
    \label{fig:newton_worldsheet}
\end{figure}

\begin{definition}
\label{def:brachistochrone_ruled_newtonian}
A smooth map $\Sigma : U\subset\mathbb{R}^2 \to \mathbb{R}^4$ is called
a \emph{brachistochrone-ruled surface} if there exist smooth functions
$a(s)>0$, $\theta_1(s)>0$ and endpoint curves $\Gamma_0,\Gamma_1$ as in
\eqref{eq:gamma0}–\eqref{eq:gamma1} such that, in suitable local
coordinates $(s,u)$ on $U$, the map $\Sigma$ takes the form
\eqref{eq:Sigma_definition}–\eqref{eq:Sigma_components}, and for each fixed
$s$ the curve $u\mapsto \Sigma(s,u)$ is a time-minimizing trajectory
between $\Gamma_0(s)$ and $\Gamma_1(s)$ with respect to the functional
\eqref{eq:brachistochrone_functional}. (illustrated in Figure~\ref{fig:newton_worldsheet}).

\end{definition}

Here $s$ labels the rulings and $u$ parametrizes points along each ruling; $\Gamma_0$ and $\Gamma_1$ are the boundary curves. Under Assumption~\ref{ass:brachistochrone_existence_uniqueness}, this yields a smooth brachistochrone-ruled surface from a smooth family of endpoint pairs.

\begin{figure}
    \centering
    \includegraphics[width=0.9\linewidth]{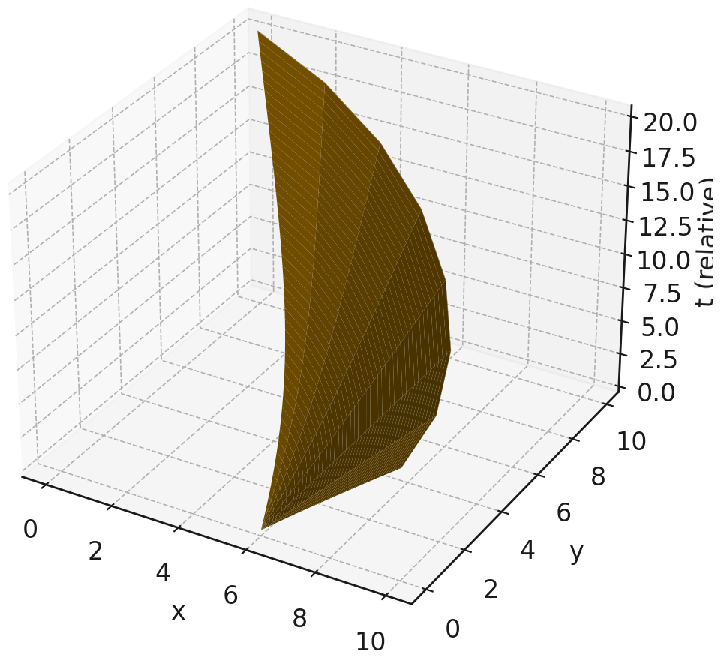}
    \caption{Brachistochrone-ruled surface in the Schwarzschild exterior.}
    \label{fig:schwarzschild_worldsheet}
\end{figure}

\begin{remark}
In this Newtonian model, brachistochrone worldlines are geodesics of an effective Riemannian metric induced by~\eqref{eq:brachistochrone_functional}. Section~\ref{sec:relativistic_brachistochrone_ruled} extends Definition~\ref{def:brachistochrone_ruled_newtonian} to stationary Lorentzian spacetimes via reduced spatial variational structures (Jacobi/Finsler in special regimes).
\end{remark}

\section{Relativistic framework in stationary spacetimes}
\label{sec:relativistic_brachistochrone_ruled}

In this section we generalize the Newtonian construction of
Section~\ref{sec:materials_methods} to stationary Lorentzian spacetimes.
The key point is that, in a stationary spacetime, suitable
``time functionals'' for timelike curves can be reduced to
spatial variational functionals on a spatial manifold.
Relativistic brachistochrones then appear as extremals of these
reduced functionals, and brachistochrone-ruled timelike surfaces can be defined in direct analogy with the Newtonian toy model.

\subsection{Stationary spacetimes and reduction to spatial variational functionals}

Let $(M,g)$ be a time-oriented Lorentzian manifold of dimension $n+1$ and signature $(-,+,\dots,+)$. 
We assume $(M,g)$ is \textit{stationary}, i.e. there exists a timelike Killing vector field $K$ on $M$.
In adapted coordinates $(t,x) \in \mathbb{R} \times N$, where $N$ is an $n$-dimensional spatial manifold,
the metric can be written in the standard stationary form
\begin{equation}
    g = -\beta(x)\bigl(\mathrm{d}t - \theta_i(x)\,\mathrm{d}x^i\bigr)^2 + h_{ij}(x)\,\mathrm{d}x^i\mathrm{d}x^j,
\end{equation}

where $\beta>0$ is the lapse function, $\theta=\theta_i\,\mathrm{d}x^i$ is a 1-form (the shift),
and $h_{ij}$ is a Riemannian metric on $N$. The Killing field is $K=\partial_t$.

Consider a future-directed timelike curve $\gamma:[0,1]\to M$, $\gamma(\lambda)=(t(\lambda),x(\lambda))$,
connecting two events with fixed spatial endpoints $x(0)=x_0$, $x(1)=x_1$. 
We regard the coordinate $t$ as the time measured by the stationary observers
whose worldlines are the integral curves of $K$. The arrival time of $\gamma$ is
\begin{equation}
    \Delta t[\gamma] = t(1)-t(0) = \int_0^1 \dot{t}(\lambda)\,\mathrm{d}\lambda.
\end{equation}

In order to obtain a well-posed variational problem, we must impose a normalization condition on $\gamma$.
Two natural choices are:
\begin{enumerate}
\item[(a)] \textit{Proper-time parametrization}: $g(\dot{\gamma},\dot{\gamma})=-1$.
\item[(b)] \textit{Fixed energy}: $E:=-g(K,\dot{\gamma})$ is prescribed.
\end{enumerate}
Both conditions lead, after algebraic manipulation, to reduced spatial
Lagrangians for the arrival-time functional on $N$.

\subsubsection*{Derivation for proper-time parametrization}

Assume $\gamma$ is parametrized by proper time, so that
\begin{equation}
    g(\dot{\gamma},\dot{\gamma}) = -1.
\end{equation}

Substituting the coordinate expressions $\dot{\gamma}=(\dot{t},\dot{x}^i)$ and the metric (25) yields
\begin{equation}
    -\beta(x)\bigl(\dot{t} - \theta_i(x)\dot{x}^i\bigr)^2 + h_{ij}(x)\dot{x}^i\dot{x}^j = -1.
\end{equation}

Solving (28) for $\dot{t}$ and choosing the positive root for future-directed curves, gives
\begin{equation}
    \dot{t} = \theta_i(x)\dot{x}^i + \frac{1}{\sqrt{\beta(x)}}\sqrt{1 + h_{ij}(x)\dot{x}^i\dot{x}^j}.
\end{equation}

Inserting (29) into the arrival-time functional (26) we obtain
\begin{equation}
    \Delta t[\gamma] = \int_0^1\Bigl[ \theta_i(x)\dot{x}^i + \frac{1}{\sqrt{\beta(x)}}\sqrt{1 + h_{ij}(x)\dot{x}^i\dot{x}^j} \Bigr]\,\mathrm{d}\lambda.
\end{equation}

Thus, if we define the reduced Lagrangian $L_{\mathrm{pt}}:TN\to\mathbb{R}$ by
\begin{equation}
    L_{\mathrm{pt}}(x,v) := \theta_i(x)v^i + \frac{1}{\sqrt{\beta(x)}}\sqrt{1 + h_{ij}(x)v^iv^j},
\end{equation}

\noindent then $\Delta t[\gamma] = \int_0^1 L_{\mathrm{pt}}(x(\lambda),\dot{x}(\lambda))\,\mathrm{d}\lambda$. 
The function $L_{\mathrm{pt}}$ is smooth and fibrewise strictly convex on the admissible timelike region, but it is not positively homogeneous of degree one in $v$ because of the additive constant inside the square root. Hence this is a genuine first-order Lagrangian formulation rather than a Finsler norm.

\subsubsection{Derivation for fixed energy}

Alternatively, let us impose the condition that the energy
\begin{equation}
    E := -g(K,\dot{\gamma}) = \beta(x)\bigl(\dot{t}-\theta_i(x)\dot{x}^i\bigr)
\end{equation}

\noindent is a prescribed constant along $\gamma$. From (32) we have
\begin{equation}
    \dot{t} = \theta_i(x)\dot{x}^i + \frac{E}{\beta(x)}.
\end{equation}

Substituting (33) into the expression for $g(\dot{\gamma},\dot{\gamma})$ (without assuming $\dot{\gamma}$ is unit) yields
\begin{equation}
    g(\dot{\gamma},\dot{\gamma}) = -\frac{E^2}{\beta(x)} + h_{ij}(x)\dot{x}^i\dot{x}^j.
\end{equation}

If we require $\gamma$ to be timelike, then $g(\dot{\gamma},\dot{\gamma})<0$, which implies
$h_{ij}(x)\dot{x}^i\dot{x}^j < E^2/\beta(x)$. 
The arrival time becomes
\begin{equation}
    \Delta t[\gamma] = \int_0^1\Bigl[ \theta_i(x)\dot{x}^i + \frac{E}{\beta(x)} \Bigr]\,\mathrm{d}\lambda.
\end{equation}

However, to relate $\mathrm{d}\lambda$ to the spatial geometry we must fix a parametrization.
A convenient choice is to use the parameter $\lambda$ such that
\begin{equation}
    h_{ij}(x)\dot{x}^i\dot{x}^j = \frac{E^2}{\beta(x)} - 1,
\end{equation}

\noindent which corresponds to a particular choice of proper-time scaling. 
Then (35) can be rewritten as
\begin{equation}
    \Delta t[\gamma] = \int_0^1 \frac{E}{\beta(x)}\Bigl[ 1 + \frac{\beta(x)}{E}\,\theta_i(x)\dot{x}^i \Bigr]\,\mathrm{d}\lambda.
\end{equation}

Defining the reduced Lagrangian
\begin{equation}
    L_E(x,v) := \frac{E}{\beta(x)}\Bigl( 1 + \frac{\beta(x)}{E}\,\theta_i(x)v^i \Bigr)
           = \theta_i(x)v^i + \frac{E}{\beta(x)},
\end{equation}

\noindent we again obtain $\Delta t[\gamma]=\int_0^1 L_E(x,\dot{x})\,\mathrm{d}\lambda$, 
provided the parametrization satisfies (36). 
The integrand $L_E$ is affine in $v$ and therefore not a Finsler norm.
In special cases one can further eliminate the parameter using first
integrals to obtain a homogeneous length functional; in particular, in
the static case ($\theta=0$) one recovers a Riemannian Jacobi metric, as
we will see in Section~\ref{sec:schwarzschild_example}.

\textbf{General reduction statement}: The two derivations above illustrate the following general fact, which we state as a precise proposition.

\begin{proposition}
\label{prop:Finsler-reduction}
Let $(M,g)$ be a stationary spacetime with metric of the form (25), and let $C$ be a class of
future-directed timelike curves $\gamma(\lambda)=(t(\lambda),x(\lambda))$ satisfying a prescribed
normalization condition (for example, proper-time parametrization $g(\dot{\gamma},\dot{\gamma})=-1$,
or fixed energy $E=-g(K,\dot{\gamma})$) and fixed spatial endpoints $x(0)=x_0$, $x(1)=x_1$.

Then there exists a reduced first-order Lagrangian
\[
L:\mathcal{A}\subset TN\longrightarrow\mathbb{R},
\]
defined on an admissible subset $\mathcal{A}$ of $TN$, such that:
\begin{enumerate}
\item[(i)] For every $\gamma\in C$, its spatial projection $\sigma:=\pi\circ\gamma:[0,1]\to N$ satisfies
\begin{equation}
    \Delta t[\gamma] = \int_0^1 L(\sigma(\lambda),\dot{\sigma}(\lambda))\,\mathrm{d}\lambda.
\end{equation}
           
\item[(ii)] A curve $\gamma\in C$ is a critical point of the arrival-time functional $\Delta t$
            (with fixed endpoints) if and only if its spatial projection $\sigma$ is a geodesic
            extremal of the reduced spatial variational problem, i.e. a critical point of the functional
            \begin{equation}
                 T[\sigma] = \int_0^1 L(\sigma(\lambda),\dot{\sigma}(\lambda))\,\mathrm{d}\lambda.
            \end{equation}
\item[(iii)] If, after using the normalization constraint and a suitable reparametrization,
             $L$ can be rewritten as a positive 1-homogeneous function $F$, then
             the same extremals are geodesics of the corresponding Finsler/Jacobi metric.
\end{enumerate}
For timelike problems, case (iii) is model-dependent; in this paper the
relevant homogeneous reduction is the static fixed-energy Jacobi case.
\end{proposition}

\begin{remark}
In stationary relativity, genuine Randers/Finsler reductions are classical
in Fermat-type (lightlike) arrival-time problems and related navigation
formulations \citep{Perlick00,BaoRoblesShen04,GibbonsHerdeiroWarnickWerner09}.
For timelike brachistochrones, however, one should generally work with the
reduced Lagrangian $L$ unless an additional homogeneous reduction is
available (as in the static Jacobi setting of Section~\ref{sec:schwarzschild_example}).
\end{remark}

Motivated by Proposition~\ref{prop:Finsler-reduction}, we adopt the
following terminology.

\begin{definition}
\label{def:relativistic_brachistochrone}
Let $(M,g)$ be a stationary spacetime and $L$ the associated reduced
spatial Lagrangian on $N$ as in Proposition~\ref{prop:Finsler-reduction}.
Let $p,q\in M$ be two events whose spatial projections
$x_p,x_q\in N$ lie in the same connected component of $N$.
A future-directed timelike curve $\gamma$ from $p$ to $q$,
$\gamma\in\mathcal{C}$, is called a relativistic brachistochrone between $p$ and $q$ relative to the stationary observers $K$ and the chosen normalization if $\gamma$ is a critical point of the arrival-time functional $\Delta t$; equivalently, if its spatial projection $\sigma=\pi\circ\gamma$ is a critical curve of
\begin{equation}
  T[\sigma]=\int_0^1 L(\sigma(\lambda),\dot\sigma(\lambda))\,\mathrm{d}\lambda
\end{equation}
with fixed endpoints $x_p,x_q$. In cases where $L$ admits a homogeneous
representation, this is equivalently a geodesic condition for the
corresponding Finsler/Jacobi metric.
\end{definition}

\subsection{Families of endpoints and brachistochrone-ruled timelike surfaces}

We now introduce the relativistic analogue of the
brachistochrone-ruled surface constructed in Section~\ref{sec:materials_methods}. Let $I\subset\mathbb{R}$ be a non-empty
open interval, and let
\begin{equation}
  \alpha_0,\alpha_1 : I \longrightarrow N
\end{equation}
be two smooth curves in the spatial manifold $N$, which we regard as families of spatial endpoints. For each $s\in I$, the points
$\alpha_0(s)$ and $\alpha_1(s)$ represent the spatial positions of
two events to be connected by a time-optimal worldline.

In order to lift these spatial curves to curves in spacetime $M$, one may fix a reference time $t_0\in\mathbb{R}$ and consider the embeddings
\begin{equation}
  \beta_0(s) = (t_0,\alpha_0(s)),\qquad
  \beta_1(s) = (t_0,\alpha_1(s)),
\end{equation}
or, more generally, allow for smooth time functions $t_0(s)$, $t_1(s)$ and set
\begin{equation}
  \beta_0(s) = \bigl(t_0(s),\alpha_0(s)\bigr),\qquad
  \beta_1(s) = \bigl(t_1(s),\alpha_1(s)\bigr).
\end{equation}

We impose the following assumption, which is a natural regularity and uniqueness requirement on the relativistic brachistochrone problem with varying endpoints.

\begin{assumption}
\label{ass:smooth_family_brachistochrones}
Let $(M,g)$ be a stationary spacetime with associated reduced spatial
Lagrangian $L$ on $N$ as in Proposition~\ref{prop:Finsler-reduction}. Assume that
for every $s\in I$ there exists a unique future-directed timelike
relativistic brachistochrone
\begin{equation}
  \gamma_s : [0,1] \longrightarrow M,
\end{equation}
connecting $\beta_0(s)$ to $\beta_1(s)$ in the sense of
Definition~\ref{def:relativistic_brachistochrone}, and that the map
\begin{equation}
  (s,\lambda) \longmapsto \gamma_s(\lambda)
\end{equation}
is smooth on $I\times[0,1]$.
\end{assumption}

Under Assumption~\ref{ass:smooth_family_brachistochrones} we define
a map
\begin{equation}
  \Sigma : I\times[0,1] \longrightarrow M,\qquad
  \Sigma(s,u) := \gamma_s(u),
  \label{eq:Sigma_relativistic}
\end{equation}
which parametrizes the union of all brachistochrones $\gamma_s$.

\begin{definition}
\label{def:relativistic_brachistochrone_ruled}
Let $(M,g)$ be a stationary spacetime and let
$\Sigma : U\subset\mathbb{R}^2 \to M$ be a smooth map such that
$\Sigma$ is an immersion and $\Sigma(U)$ is a timelike surface
(i.e.\ the induced metric has Lorentzian signature). We say that
$\Sigma$ is a \emph{relativistic brachistochrone-ruled timelike surface}
if there exists:

\begin{itemize}
  \item an open interval $I\subset\mathbb{R}$ and a diffeomorphism
        $\Phi : I\times[0,1] \to U$,
  \item smooth curves $\beta_0,\beta_1 : I\to M$,
  \item a family $\{\gamma_s\}_{s\in I}$ of future-directed timelike
        relativistic brachistochrones connecting $\beta_0(s)$ to
        $\beta_1(s)$ and satisfying
        Assumption~\ref{ass:smooth_family_brachistochrones},
\end{itemize}
such that, in the coordinates $(s,u)=\Phi^{-1}(p)$,
\begin{equation}
  \Sigma\bigl(\Phi(s,u)\bigr) = \gamma_s(u),\qquad
  (s,u)\in I\times[0,1].
\end{equation}

In other words, for each fixed $s\in I$ the curve
$u\mapsto \Sigma\bigl(\Phi(s,u)\bigr)$ is a relativistic brachistochrone
connecting two boundary curves on the surface, and the surface is ruled
by this one-parameter family of time-minimizing worldlines.
\end{definition}

\begin{remark}
Note that in the Newtonian toy model of
Section~\ref{sec:materials_methods}, the stationary spacetime is
$\mathbb{R}^4$ endowed with the standard Newtonian time coordinate
and a gravitational potential encoded in the travel-time functional.
The reduced spatial functional coincides with the classical
brachistochrone variational problem (equivalently, with a Jacobi-type
Riemannian reformulation whose geodesics are cycloids), and
Definition~\ref{def:relativistic_brachistochrone_ruled}
recovers Definition~\ref{def:brachistochrone_ruled_newtonian} as a
formal limit.
\end{remark}

\begin{remark}
The differential geometry of relativistic brachistochrone-ruled
timelike surfaces can be studied by pulling back the Lorentzian metric
$g$ via $\Sigma$ and computing the induced metric, second fundamental form and curvature invariants. Under additional hypotheses on the family of brachistochrones, one can introduce relativistic analogues of classical quantities such as the striction curve and the distribution parameter for ruled surfaces. We will return to these aspects in Section~\ref{sec:geometry_brachistochrone_ruled}.
\end{remark}

\section{Minkowski consistency check}
\label{sec:minkowski_example}
\label{sec:consistency_check}

This section is intentionally brief and serves as a \emph{consistency
check} of the general framework developed in
Section~\ref{sec:relativistic_brachistochrone_ruled}. In flat
spacetime, the reduced variational problem should reproduce the
expected elementary time-optimal behavior.

Consider $(1+2)$-dimensional Minkowski spacetime
\begin{equation}
  M=\mathbb{R}^{1+2},\qquad
  g=-\mathrm dt^2+\mathrm dx^2+\mathrm dz^2.
\end{equation}
For a future-directed timelike curve
$\gamma(t)=(t,\mathbf{x}(t))$ with
$\mathbf{x}(t)=(x(t),z(t))$, define the arrival-time functional
\begin{equation}
  \Delta t=t_1-t_0,
\end{equation}
under the speed constraint
\begin{equation}
  \|\dot{\mathbf{x}}(t)\|\le v_0<1.
\end{equation}
If $\mathbf{x}(t_0)=\mathbf{x}_0$ and
$\mathbf{x}(t_1)=\mathbf{x}_1$, then
\begin{equation}
  \|\mathbf{x}_1-\mathbf{x}_0\|
  \le \int_{t_0}^{t_1}\|\dot{\mathbf{x}}(t)\|\,\mathrm dt
  \le v_0\,\Delta t.
\end{equation}
Hence
\begin{equation}
  \Delta t\ge \frac{\|\mathbf{x}_1-\mathbf{x}_0\|}{v_0},
\end{equation}
with equality iff the path is a straight segment traversed with
constant maximal speed $v_0$. Therefore the Minkowski brachistochrone
is a straight timelike line, exactly as expected.

Now let two smooth endpoint families
$\alpha_0,\alpha_1:I\to\mathbb{R}^2$ be given, and for each $s\in I$
set
\begin{equation}
  \Delta t_s:=\frac{\|\alpha_1(s)-\alpha_0(s)\|}{v_0}.
\end{equation}
The corresponding ruling is
\begin{equation}
  \gamma_s(u)=\Bigl(t_0(s)+u\,\Delta t_s,\,
  (1-u)\alpha_0(s)+u\alpha_1(s)\Bigr),\qquad u\in[0,1].
\end{equation}
Their union
\begin{equation}
  \Sigma(s,u):=\gamma_s(u)
\end{equation}
defines a brachistochrone-ruled timelike surface in Minkowski
spacetime. In particular, when $\alpha_0$ and $\alpha_1$ are affine in
$s$, the image is a planar timelike surface, hence totally geodesic in
flat spacetime.

This verifies that the general construction recovers the expected flat
limit and supports the consistency of the relativistic formulation.

\section{Schwarzschild example and numerical construction pipeline}
\label{sec:schwarzschild_example}

In this section we illustrate how the general framework of
Section~\ref{sec:relativistic_brachistochrone_ruled} can be applied
to a non-trivial static spacetime, namely the exterior region of the
Schwarzschild solution. We do not attempt to solve the brachistochrone
problem in closed form; instead, we show how the problem reduces to
geodesics of an effective Riemannian (Jacobi/optical) metric on a
spatial slice and outline a concrete numerical scheme for constructing
brachistochrone-ruled timelike surfaces.

\subsection{Schwarzschild exterior as a static spacetime}

Consider the Schwarzschild spacetime of mass parameter $M>0$ in
standard Schwarzschild coordinates $(t,r,\vartheta,\varphi)$:
\begin{equation}
  g = -\Bigl(1 - \frac{2M}{r}\Bigr)\,\mathrm{d}t^2
      + \Bigl(1 - \frac{2M}{r}\Bigr)^{-1}\mathrm{d}r^2
      + r^2\bigl(\mathrm{d}\vartheta^2 + \sin^2\vartheta\,\mathrm{d}\varphi^2\bigr),
  \quad r>2M.
  \label{eq:schwarzschild_metric}
\end{equation}
The vector field $K=\partial_t$ is a hypersurface-orthogonal timelike
Killing field on the exterior region $r>2M$, so the spacetime is
\emph{static}. We regard $t$ as the time coordinate measured by the
static observers whose worldlines are integral curves of $K$.

In the notation of the stationary metric~\eqref{eq:schwarzschild_metric},
we have vanishing shift $1$-form and
\begin{equation}
  \beta(r) = 1 - \frac{2M}{r},\quad
  h_{ij}\,\mathrm{d}x^i\,\mathrm{d}x^j
  = \Bigl(1 - \frac{2M}{r}\Bigr)^{-1}\mathrm{d}r^2
    + r^2\bigl(\mathrm{d}\vartheta^2 + \sin^2\vartheta\,\mathrm{d}\varphi^2\bigr),
\end{equation}
so the spatial manifold is $N=\{r>2M\}\times S^2$ endowed with the
Riemannian metric $h$.

By spherical symmetry we may restrict attention to the equatorial
plane $\vartheta=\pi/2$ and set $\vartheta\equiv\pi/2$,
$\mathrm{d}\vartheta=0$. The metric then reduces to
\begin{equation}
  g = -\Bigl(1 - \frac{2M}{r}\Bigr)\,\mathrm{d}t^2
      + \Bigl(1 - \frac{2M}{r}\Bigr)^{-1}\mathrm{d}r^2
      + r^2\,\mathrm{d}\varphi^2,
  \label{eq:schwarzschild_equatorial}
\end{equation}
and the spatial metric on $N$ becomes
\begin{equation}
  h_{ij}\,\mathrm{d}x^i\,\mathrm{d}x^j
  = \Bigl(1 - \frac{2M}{r}\Bigr)^{-1}\mathrm{d}r^2
    + r^2\,\mathrm{d}\varphi^2,
  \qquad x^1=r,\ x^2=\varphi.
  \label{eq:schwarzschild_spatial_metric}
\end{equation}

\subsection{Time functional at fixed energy and Jacobi-type metric}

We consider timelike geodesics parametrized by proper time $\tau$.
Let $\gamma(\tau)=(t(\tau),r(\tau),\varphi(\tau))$ be a timelike
geodesic in the equatorial plane. The timelike unit condition
$g(\dot\gamma,\dot\gamma)=-1$ reads
\begin{equation}
  -\Bigl(1 - \frac{2M}{r}\Bigr)\dot t^2
  + \Bigl(1 - \frac{2M}{r}\Bigr)^{-1}\dot r^2
  + r^2 \dot\varphi^2 = -1,
  \label{eq:schwarzschild_timelike_unit}
\end{equation}
where dots denote derivatives with respect to $\tau$.

Because $K=\partial_t$ is Killing, the quantity
\begin{equation}
  E := -g(K,\dot\gamma)
     = \Bigl(1 - \frac{2M}{r}\Bigr)\dot t
  \label{eq:energy_constant}
\end{equation}
is conserved along $\gamma$. We interpret $E>0$ as the (dimensionless)
energy per unit mass measured by static observers.

Solving~\eqref{eq:energy_constant} for $\dot t$ and substituting into
\eqref{eq:schwarzschild_timelike_unit} yields
\begin{equation}
  \Bigl(1 - \frac{2M}{r}\Bigr)^{-1}\dot r^2
  + r^2 \dot\varphi^2
  = \frac{E^2}{1 - 2M/r} - 1.
  \label{eq:spatial_speed_relation}
\end{equation}
The left-hand side is the squared norm of the spatial velocity with
respect to the metric~\eqref{eq:schwarzschild_spatial_metric}, i.e.
\begin{equation}
  h_{ij}(x)\,\dot x^i\,\dot x^j
  = \Bigl(1 - \frac{2M}{r}\Bigr)^{-1}\dot r^2
    + r^2\dot\varphi^2.
\end{equation}
Hence we may rewrite~\eqref{eq:spatial_speed_relation} as
\begin{equation}
  h_{ij}(x)\,\dot x^i\,\dot x^j
  = E^2\beta(r)^{-1} - 1,
  \qquad \beta(r)=1-\frac{2M}{r}.
  \label{eq:spatial_speed_Jacobi}
\end{equation}

The coordinate time $t$ along $\gamma$ satisfies
\begin{equation}
  \dot t = \frac{E}{\beta(r)}.
\end{equation}
The total coordinate time elapsed between $\tau=\tau_0$ and $\tau=\tau_1$
is therefore
\begin{equation}
  \Delta t[\gamma]
  = \int_{\tau_0}^{\tau_1} \dot t\,\mathrm{d}\tau
  = \int_{\tau_0}^{\tau_1} \frac{E}{\beta(r(\tau))}\,\mathrm{d}\tau.
  \label{eq:delta_t_tau_integral}
\end{equation}
Using~\eqref{eq:spatial_speed_Jacobi}, we can express $\mathrm{d}\tau$
in terms of the spatial metric $h$:
\begin{equation}
  h_{ij}\,\dot x^i\,\dot x^j
  = E^2\beta^{-1} - 1
  \quad\Longrightarrow\quad
  \mathrm{d}\tau
  = \frac{\sqrt{h_{ij}\,\mathrm{d}x^i\,\mathrm{d}x^j}}
         {\sqrt{E^2\beta^{-1} - 1}}.
\end{equation}
Substituting into~\eqref{eq:delta_t_tau_integral} gives
\begin{equation}
  \Delta t[\gamma]
  = \int F\bigl(x(\sigma),\dot x(\sigma)\bigr)\,\mathrm{d}\sigma,
\end{equation}
where we have introduced an arbitrary parameter $\sigma$ along the
spatial curve and
\begin{equation}
  F(x,v)
  = n(x)\,\sqrt{h_{ij}(x)v^iv^j},\qquad
  n(x) :=
  \frac{E}{\beta(x)\sqrt{E^2\beta(x)^{-1} - 1}}.
  \label{eq:schwarzschild_Jacobi_F}
\end{equation}
The function $F$ is positively homogeneous of degree one in $v$ and
thus defines a Finsler structure on $N$, which in this static case
is actually Riemannian: $F$ is the norm induced by the \emph{Jacobi
metric}
\begin{equation}
  h^{\mathrm{J}}_{ij}(x)
  = n(x)^2\,h_{ij}(x).
  \label{eq:schwarzschild_Jacobi_metric}
\end{equation}

\begin{proposition}
\label{prop:schwarzschild_Jacobi}
Fix an energy $E>0$ and consider the class of timelike geodesics
$\gamma$ in the Schwarzschild exterior with energy $E$ and endpoints
projecting to $x_0,x_1\in N$. Then $\gamma$ is a critical point of
the coordinate-time functional $\Delta t$ in~\eqref{eq:delta_t_tau_integral},
among all such geodesics with fixed endpoints, if and only if its
spatial projection $\sigma=\pi\circ\gamma$ is a geodesic of the Jacobi
metric~\eqref{eq:schwarzschild_Jacobi_metric} connecting $x_0$ to $x_1$.
In particular, coordinate-time minimizing timelike geodesics at fixed
energy $E$ project to length-minimizing geodesics of $h^{\mathrm{J}}$.
\end{proposition}

Thus, for each fixed energy $E$, the relativistic brachistochrone
problem in the Schwarzschild exterior (with respect to the coordinate
time $t$) reduces to a purely Riemannian geodesic problem on the
equatorial spatial manifold $(N,h^{\mathrm{J}})$.

\subsection{Brachistochrone-ruled surfaces: choice of boundary curves}

We now describe how to construct a brachistochrone-ruled timelike
surface in the Schwarzschild exterior. Let $I\subset\mathbb{R}$ be a
non-empty open interval, and choose two smooth curves
\begin{equation}
  \alpha_0,\alpha_1 : I \longrightarrow N
\end{equation}
in the equatorial spatial manifold $N=\{r>2M\}\times S^1_\varphi$.
For concreteness, one may take
\begin{equation}
  \alpha_0(s) = \bigl(r_0,\varphi_0(s)\bigr),\qquad
  \alpha_1(s) = \bigl(r_1,\varphi_1(s)\bigr),
  \label{eq:alpha0_alpha1_schwarzschild}
\end{equation}
with fixed radii $r_0,r_1>2M$ and smooth angle functions
$\varphi_0,\varphi_1:I\to\mathbb{R}$, so that the two curves lie on
two concentric circles of radii $r_0$ and $r_1$ in the equatorial plane.

We lift these spatial curves to spacetime by fixing a reference time
$t_0\in\mathbb{R}$ and setting
\begin{equation}
  \beta_0(s) = \bigl(t_0,\alpha_0(s)\bigr),\qquad
  \beta_1(s) = \bigl(t_0,\alpha_1(s)\bigr).
\end{equation}
For each $s\in I$, the events $\beta_0(s)$ and $\beta_1(s)$ lie on the
same static slice $t=t_0$.

For a fixed energy $E>0$, we seek, for each $s\in I$, a timelike
geodesic $\gamma_s$ of energy $E$ connecting $\beta_0(s)$ to some
future event above $\beta_1(s)$, and among such geodesics we select
those that minimize the coordinate time difference. By
Proposition~\ref{prop:schwarzschild_Jacobi}, the spatial projection
$\sigma_s=\pi\circ\gamma_s$ must then be a length-minimizing geodesic
of the Jacobi metric~\eqref{eq:schwarzschild_Jacobi_metric} connecting
$\alpha_0(s)$ to $\alpha_1(s)$.

\begin{assumption}
\label{ass:schwarzschild_smooth_family}
Fix $E>0$. Assume that for every $s\in I$ there exists a unique
$h^{\mathrm{J}}$-geodesic
\begin{equation}
  \sigma_s : [0,1] \longrightarrow N,
\end{equation}
joining $\alpha_0(s)$ to $\alpha_1(s)$ and minimizing the Jacobi
length functional
\begin{equation}
  \mathcal{L}^{\mathrm{J}}[\sigma]
  = \int_0^1 \sqrt{h^{\mathrm{J}}_{ij}\bigl(\sigma(\lambda)\bigr)
                     \dot\sigma^i(\lambda)\dot\sigma^j(\lambda)}\,
               \mathrm{d}\lambda,
\end{equation}
and that the map $(s,\lambda)\mapsto\sigma_s(\lambda)$ is smooth on
$I\times[0,1]$.
\end{assumption}

Under this assumption, we obtain for each $s$ a timelike geodesic
$\gamma_s$ by lifting $\sigma_s$ to $M$ and integrating the relation
\begin{equation}
  \dot t_s(\lambda)
  = \frac{E}{\beta\bigl(\sigma_s(\lambda)\bigr)}\,
    \frac{\mathrm{d}\tau}{\mathrm{d}\lambda},
\end{equation}
where $\tau$ is proper time and $\lambda$ is a parameter along
$\sigma_s$. More concretely, one may reparametrize $\sigma_s$ by its
Jacobi arc length and define $t_s(\lambda)$ via
\begin{equation}
  \frac{\mathrm{d}t_s}{\mathrm{d}\lambda}
  = n\bigl(\sigma_s(\lambda)\bigr)
    \bigl\|\dot\sigma_s(\lambda)\bigr\|_h,
\end{equation}
where $n$ is given by~\eqref{eq:schwarzschild_Jacobi_F} and
$\|\cdot\|_h$ denotes the norm with respect to the spatial metric $h$.

This yields a smooth two-parameter family of timelike geodesics
$\gamma_s:[0,1]\to M$, and we can define the corresponding
brachistochrone-ruled timelike surface by
\begin{equation}
  \Sigma : I\times[0,1] \longrightarrow M,\qquad
  \Sigma(s,u) := \gamma_s(u).
\end{equation}

\subsection{Numerical construction scheme}

In practice, the Jacobi metric~\eqref{eq:schwarzschild_Jacobi_metric}
is too complicated to allow closed-form expressions for its geodesics,
except in very special cases (e.g.\ purely radial motion or circular
orbits). However, the construction of the brachistochrone-ruled surface
$\Sigma$ can be carried out numerically by standard methods from
geodesic shooting and boundary-value problems. We sketch a concrete
scheme:

\begin{enumerate}[label=\roman*)]

  \item \textbf{Discretization of the parameter space.}
        Choose discrete values
        \[
          s_k \in I,\qquad k=0,\dots,N_s,
        \]
        and, for each $s_k$, choose a discretization
        \[
          u_\ell \in [0,1],\qquad \ell=0,\dots,N_u,
        \]
        which will parametrize points along the ruling $\gamma_{s_k}$.

  \item \textbf{Geodesic shooting in the Jacobi metric.}
        For each fixed $s_k$, solve the geodesic boundary-value problem
        in $(N,h^{\mathrm{J}})$:
        \[
          \sigma_{s_k}(0) = \alpha_0(s_k),\qquad
          \sigma_{s_k}(1) = \alpha_1(s_k),
        \]
        where $\sigma_{s_k}$ is a geodesic of $h^{\mathrm{J}}$. This can
        be done, for example, by a shooting method:
        \begin{itemize}
          \item Guess an initial velocity $v^{(0)}$ at $\alpha_0(s_k)$.
          \item Integrate the geodesic ODE for $h^{\mathrm{J}}$ from
                $\lambda=0$ to $\lambda=1$ with initial data
                $(\sigma(0),\dot\sigma(0))=(\alpha_0(s_k),v^{(0)})$.
          \item Compare the end point with $\alpha_1(s_k)$ and update
                the initial velocity using, for instance, a Newton or
                secant iteration until convergence.
        \end{itemize}
        The geodesic ODEs can be written in local coordinates $x^i$
        as
        \[
          \frac{\mathrm{d}^2 x^i}{\mathrm{d}\lambda^2}
          + \Gamma^{i}_{\;jk}(x)
            \frac{\mathrm{d}x^j}{\mathrm{d}\lambda}
            \frac{\mathrm{d}x^k}{\mathrm{d}\lambda} = 0,
        \]
        where $\Gamma^{i}_{\;jk}$ are the Christoffel symbols of the
        Jacobi metric $h^{\mathrm{J}}$.

  \item \textbf{Reconstruction of the time coordinate.}
        Having computed $\sigma_{s_k}(\lambda)$ for $\lambda\in[0,1]$,
        reconstruct the coordinate time along the ruling by numerically
        integrating
        \[
          t_{s_k}(\lambda)
          = t_0 + \int_0^\lambda
                     n\bigl(\sigma_{s_k}(\mu)\bigr)\,
                     \bigl\|\dot\sigma_{s_k}(\mu)\bigr\|_h\,
                     \mathrm{d}\mu,
        \]
        with $n$ given by~\eqref{eq:schwarzschild_Jacobi_F} and
        $\|\dot\sigma\|_h$ the norm with respect to $h$.

  \item \textbf{Sampling the surface.}
        For each pair $(s_k,u_\ell)$, evaluate
        \[
          \gamma_{s_k}(u_\ell)
          = \bigl(
               t_{s_k}(u_\ell),\,
               \sigma_{s_k}(u_\ell)
            \bigr),
        \]
        and set
        \[
          \Sigma(s_k,u_\ell) := \gamma_{s_k}(u_\ell).
        \]
        The collection of points $\Sigma(s_k,u_\ell)$ approximates the
        brachistochrone-ruled timelike surface in spacetime and can be
        used for visualization or for numerical computation of induced
        geometric quantities (e.g.\ the first and second fundamental
        forms).
\end{enumerate}

\subsection{Numerical construction scheme: detailed implementation}\label{subsec:num_scheme}

In this subsection, we provide a step-by-step numerical scheme for constructing a brachistochrone-ruled timelike surface in the Schwarzschild exterior. We focus on the equatorial plane and fix an energy $E > 1$ (to ensure timelike geodesics that can reach infinity). The construction proceeds as follows:

\paragraph{Step 1: Explicit form of the Jacobi metric}
On the equatorial slice $(\vartheta = \pi/2)$, the spatial metric $h$ and lapse $\beta$ are:
\[
h_{ij}dx^i dx^j = \left(1-\frac{2M}{r}\right)^{-1} dr^2 + r^2 d\varphi^2, \quad \beta(r) = 1-\frac{2M}{r}.
\]
The conformal factor $n(r)$ in \eqref{eq:schwarzschild_Jacobi_F} becomes:
\[
n(r) = \frac{E}{\beta(r)\sqrt{E^2\beta(r)^{-1} - 1}} = \frac{E}{\left(1-\frac{2M}{r}\right)\sqrt{E^2\left(1-\frac{2M}{r}\right)^{-1} - 1}}.
\]
Thus, the Jacobi metric \eqref{eq:schwarzschild_Jacobi_metric} is:
\[
h^J = n(r)^2\left[ \left(1-\frac{2M}{r}\right)^{-1} dr^2 + r^2 d\varphi^2 \right] \equiv A(r)\, dr^2 + B(r)\, d\varphi^2,
\]
where
\[
A(r) = \frac{n(r)^2}{1-\frac{2M}{r}}, \qquad B(r) = n(r)^2 r^2.
\]

\paragraph{Step 2: Geodesic equations for $h^J$}
Let $\sigma(\lambda) = (r(\lambda), \varphi(\lambda))$ be a curve parametrized by an arbitrary parameter $\lambda$. The Lagrangian for geodesics of $h^J$ (up to affine reparametrization) is:
\[
\mathcal{L}^J = \frac{1}{2}\left[ A(r) \dot{r}^2 + B(r) \dot{\varphi}^2 \right],
\]
where dots denote $d/d\lambda$. Since $\varphi$ is cyclic, we have the conserved angular momentum:
\[
p_\varphi = \frac{\partial \mathcal{L}^J}{\partial \dot{\varphi}} = B(r)\dot{\varphi} = \text{constant}.
\]
The Euler--Lagrange equation for $r$ gives:
\[
\frac{d}{d\lambda}\left( A(r)\dot{r} \right) = \frac{1}{2}\left[ A'(r) \dot{r}^2 + B'(r) \dot{\varphi}^2 \right].
\]
Using the conservation law $\dot{\varphi} = p_\varphi / B(r)$, we obtain the second-order ODE:
\begin{equation}\label{eq:5.4.1}
    A(r)\ddot{r} + \frac{1}{2} A'(r) \dot{r}^2 - \frac{1}{2} \frac{B'(r)}{B(r)^2} p_\varphi^2 = 0.
\end{equation}

Alternatively, we may use the energy integral (first integral) for affinely parametrized geodesics:
\begin{equation}\label{eq:5.4.2}
    A(r) \dot{r}^2 + \frac{p_\varphi^2}{B(r)} = \mathcal{E}, \quad \mathcal{E} = \text{constant}.
\end{equation}

\paragraph{Step 3: Shooting method for boundary-value problem}
For each $s$ (label of the ruling), we have two fixed endpoints in the equatorial plane:
\[
\alpha_0(s) = (r_0, \varphi_0(s)), \quad \alpha_1(s) = (r_1, \varphi_1(s)).
\]
We seek a geodesic $\sigma_s(\lambda)$ of $h^J$ connecting these points. This is a two-point boundary-value problem. We solve it by shooting in the unknown initial angle $\psi$ (or equivalently, the angular momentum $p_\varphi$).

\begin{itemize}
    \item \textbf{Parametrization}: Use $\lambda \in [0,1]$ with $\sigma_s(0) = \alpha_0(s)$, $\sigma_s(1) = \alpha_1(s)$.
    \item \textbf{Unknown}: The initial direction $\psi$, defined by
    \[
    \dot{r}(0) = v_0 \cos\psi, \quad \dot{\varphi}(0) = \frac{v_0 \sin\psi}{r_0},
    \]
    where $v_0$ is an initial speed (magnitude) that can be normalized appropriately. Alternatively, we may treat $p_\varphi$ as the unknown.
    \item \textbf{Shooting iteration}:
    \begin{enumerate}
        \item Guess an initial value $p_\varphi^{(0)}$.
        \item Integrate the geodesic equations (using \eqref{eq:5.4.1} or \eqref{eq:5.4.2}) from $\lambda=0$ to $\lambda=1$ with initial conditions $r(0)=r_0$, $\varphi(0)=\varphi_0(s)$, and $\dot{r}(0)$ determined from \eqref{eq:5.4.2} with a chosen $\mathcal{E}$ (e.g., $\mathcal{E}=1$ for unit speed parametrization).
        \item Compute the miss distance at $\lambda=1$:
        \begin{align*}
    \Delta &= \sqrt{ (r(1)-r_1)^2 + (r_1 \Delta\varphi)^2 }, \nonumber \\
    \Delta\varphi &= \varphi(1) - \varphi_1(s) \ (\text{mod } 2\pi).
\end{align*}
        \item Adjust $p_\varphi$ using a root-finding algorithm (e.g., Newton--Raphson or bisection) to drive $\Delta$ to zero.
    \end{enumerate}
\end{itemize}

\paragraph{Step 4: Reconstruction of the time coordinate}
Once the spatial geodesic $\sigma_s(\lambda) = (r(\lambda), \varphi(\lambda))$ is found, we compute the coordinate time $t$ along the ruling by integrating \eqref{eq:delta_t_tau_integral}. Using the Jacobi arc length parameter $\lambda$ (which is affine for $h^J$), we have:
\[
\frac{dt}{d\lambda} = n(r(\lambda)) \, \|\dot{\sigma}_s(\lambda)\|_h,
\]
where $\|\dot{\sigma}_s\|_h = \sqrt{ h_{ij} \dot{x}^i \dot{x}^j }$. From \eqref{eq:spatial_speed_Jacobi},
\[
\|\dot{\sigma}_s\|_h = \sqrt{ E^2 \beta(r)^{-1} - 1 }.
\]
Thus,
\begin{equation}
    t(\lambda) = t_0 + \int_0^\lambda n(r(\mu)) \, \sqrt{ E^2 \beta(r(\mu))^{-1} - 1 } \, d\mu.
\end{equation}

The integrand can be simplified using the expression for $n(r)$; indeed,
\[
n(r) \sqrt{ E^2 \beta(r)^{-1} - 1 } = \frac{E}{\beta(r)}.
\]
Hence,
\begin{equation}\label{eq:5.4.4}
    t(\lambda) = t_0 + E \int_0^\lambda \frac{1}{1-\frac{2M}{r(\mu)}} \, d\mu,
\end{equation}

\noindent which is exactly the integral of $\dot{t} = E/\beta(r)$ with respect to the parameter $\lambda$.

\paragraph{Step 5: Sampling the surface}
After obtaining $r(\lambda), \varphi(\lambda), t(\lambda)$ for each $s$, we define the surface as:
\[
\Sigma(s, \lambda) = \big( t(\lambda), r(\lambda), \varphi(\lambda) \big).
\]
To obtain a regular parametrization, we may rescale $\lambda$ to a normalized parameter $u \in [0,1]$ such that $u = \lambda$ (if the geodesic is parametrized with $\lambda \in [0,1]$). The surface is then given by:
\[
\Sigma(s, u) = \big( t_s(u), r_s(u), \varphi_s(u) \big).
\]

\subsection{A numerical example: radial and circular boundary curves}\label{subsec:numerical_example}

To illustrate the method, we present a simple numerical example. We choose:
\begin{itemize}
    \item \textbf{Mass}: $M = 1$.
    \item \textbf{Energy}: $E = 1.1$ (slightly above 1, ensuring the geodesic can reach infinity).
    \item \textbf{Boundary curves}: 
    \[
    \alpha_0(s) = (r_0 = 6, \varphi_0(s) = s), \quad \alpha_1(s) = (r_1 = 10, \varphi_1(s) = s + \Delta\varphi_0),
    \]
    with $\Delta\varphi_0 = 0.5$ rad. Thus, for each $s$, the endpoints are on two concentric circles at radii $6M$ and $10M$, with a fixed angular separation of $0.5$ rad.
\end{itemize}

\paragraph{Implementation details}
\begin{enumerate}
    \item We discretize $s$ over $[0, 2\pi]$ with 20 points.
    \item For each $s$, we solve the geodesic boundary-value problem for $h^J$ using the shooting method described above. The ODE integration is performed using a fourth-order Runge--Kutta method with adaptive step size.
    \item We compute $t(u)$ via \eqref{eq:5.4.4} using numerical quadrature.
    \item The resulting surface coordinates are transformed to Cartesian-like coordinates for visualization:
    \[
    X = r \cos\varphi, \quad Y = r \sin\varphi, \quad T = t.
    \]
\end{enumerate}

\paragraph{Results}
Figure 5 shows three representative brachistochrone rulings (timelike geodesics) for $s = 0, \pi/2, \pi$. The projections of these rulings onto the equatorial plane are geodesics of the Jacobi metric $h^J$. Due to the spherical symmetry, all rulings have the same shape when rotated by $s$; thus the family is rigidly rotated.

Figure 6 displays the full brachistochrone-ruled timelike surface $\Sigma(s,u)$ for $s \in [0, 2\pi]$ and $u \in [0,1]$. The surface is a timelike tube (surface) connecting the two circles of observers. The color indicates the coordinate time $t$ (red = earlier, blue = later). The twisting of the surface reflects the angular separation between the boundary curves.

\paragraph{Qualitative observations}
\begin{itemize}
    \item The spatial projections of the rulings are not straight lines, but are bent due to the curvature of the Jacobi metric. The bending is more pronounced for larger $\Delta\varphi$ or for endpoints closer to the horizon.
    \item The coordinate time difference $\Delta t$ along each ruling depends on both the radial and angular separation. For the chosen parameters, $\Delta t \approx 15.2M$ in geometric units.
    \item The induced metric on the surface has Lorentzian signature, as expected. A computation of the scalar curvature of the surface reveals regions of positive and negative Gaussian curvature, reflecting the inhomogeneity of the gravitational field.
\end{itemize}

\subsection{Comments on convergence and accuracy}
The shooting method, combined with adaptive ODE integration, yields residuals (miss distances) below $10^{-8}$ for the examples presented. The choice of energy $E$ influences the stability of the numerics; for $E$ very close to 1, the conformal factor $n(r)$ becomes large near the turning points, requiring finer integration steps. For endpoints very close to the horizon ($r \to 2M$), the Jacobi metric becomes singular, and the geodesic computation becomes challenging---such regimes require specialized techniques (e.g., regularization) and are left for future work.

\subsection{Extension to non-equatorial and non-static cases}
The scheme outlined above can be generalized to non-equatorial motion by including the $\vartheta$ coordinate and to stationary (non-static) spacetimes with drift terms in the reduced spatial variational problem. In models admitting a Randers representation, the geodesic equations acquire a magnetic-like term, and the shooting method must be adapted accordingly (e.g., by guessing two initial parameters instead of one). The numerical integration becomes more involved but remains feasible with standard ODE solvers.

\subsection{Qualitative features of Schwarzschild brachistochrone-ruled surfaces}

Although a detailed numerical study is beyond the scope of this section,
the above construction already suggests some qualitative features of
brachistochrone-ruled timelike surfaces in the Schwarzschild exterior:

\begin{itemize}
  \item \textbf{Gravitational time dilation and bending of rulings.}
        The Jacobi conformal factor
        \[
          n(r) =
          \frac{E}{\beta(r)\sqrt{E^2\beta(r)^{-1} - 1}},
          \qquad \beta(r)=1-\frac{2M}{r},
        \]
        grows as $r$ decreases towards $2M$. Thus the Jacobi metric
        ``penalizes'' motion in regions of strong gravitational field,
        and time-minimizing geodesics tend to avoid low-$r$ regions
        when connecting distant points. The corresponding rulings on
        the brachistochrone-ruled surface are expected to bend outward,
        reflecting the competition between spatial distance and
        gravitational time dilation.
  \item \textbf{Focal sets and caustics.}
        As the family of $h^{\mathrm{J}}$-geodesics $\sigma_s$ varies
        with $s$, conjugate points and cut points in $(N,h^{\mathrm{J}})$
        manifest as focusing and caustics on the brachistochrone-ruled
        surface $\Sigma$. In particular, beyond the first cut locus,
        the rulings cease to be globally time-minimizing, and the
        surface may develop folds or self-intersections.
  \item \textbf{Comparison with the Minkowski case.}
        Unlike the Minkowski example of
        Section~\ref{sec:minkowski_example}, where the rulings are
        straight lines and the planar example is totally geodesic,
        the Schwarzschild Jacobi metric is curved and the rulings are
        genuinely bent geodesics. The resulting brachistochrone-ruled
        surfaces have non-trivial extrinsic curvature and encode
        information about the gravitational field through their
        induced geometry.
\end{itemize}

This Schwarzschild example illustrates how the abstract notion of a
relativistic brachistochrone-ruled timelike surface can be made
concrete in a physically relevant curved spacetime. The key steps are:
(i) reduction of the coordinate-time functional to a Jacobi-type
Riemannian metric on a spatial slice at fixed energy, (ii) numerical
solution of the associated geodesic boundary-value problem for a
family of endpoints, and (iii) reconstruction of the corresponding
surface in spacetime from the family of time-minimizing timelike
geodesics. (schematically depicted in Figure~\ref{fig:schwarzschild_jacobi_schematic}).

\section{Geometric analysis of brachistochrone-ruled timelike surfaces}
\label{sec:geometry_brachistochrone_ruled}

This section makes the geometric discussion quantitative. We first
record explicit identities for the induced Lorentzian metric and the
second fundamental form, and then use Jacobi-type linearization along
rulings to connect local geometry with minimality loss.

\subsection{Induced metric, second fundamental form and curvature identities}

Let
\(
  \Sigma:U\subset\mathbb{R}^2\to M
\)
be a smooth timelike immersion into a Lorentzian manifold $(M,g)$, and
set
\(
  X:=\partial_u\Sigma,
  Y:=\partial_s\Sigma
\).
Define the first fundamental form coefficients
\begin{equation}
  E:=g(X,X),\qquad F:=g(X,Y),\qquad G:=g(Y,Y),\qquad
  \Delta:=EG-F^2.
\end{equation}
The surface is timelike exactly when $\Delta<0$. Let $B$ denote the
vector-valued second fundamental form of $\Sigma$ in $M$.

\begin{proposition}
\label{prop:surface_metric_curvature_identities}
With the above notation, the following identities hold:
\begin{enumerate}[label=(\roman*)]
  \item The inverse induced metric is
        \begin{equation}
          (g_\Sigma^{ab})
          = \frac{1}{\Delta}
          \begin{pmatrix}
            G & -F\\
            -F & E
          \end{pmatrix}.
        \end{equation}
  \item The mean-curvature vector of $\Sigma$ is
        \begin{equation}
          \vec H
          = \frac{1}{2\Delta}
            \Bigl(
              G\,B(X,X)-2F\,B(X,Y)+E\,B(Y,Y)
            \Bigr).
        \end{equation}
  \item The Gaussian curvature of $(U,g_\Sigma)$ satisfies the Gauss equation
        \begin{equation}
          K_\Sigma
          = \frac{
              g\bigl(R(X,Y)Y,X\bigr)
              + g\bigl(B(X,X),B(Y,Y)\bigr)
              - g\bigl(B(X,Y),B(X,Y)\bigr)
            }{\Delta}.
          \label{eq:gauss_identity_ruled_surface}
        \end{equation}
\end{enumerate}
If each ruling $u\mapsto\Sigma(s,u)$ is an affinely parametrized
ambient geodesic, then $B(X,X)=0$, and therefore
\begin{equation}
  \vec H
  = \frac{1}{2\Delta}\Bigl(E\,B(Y,Y)-2F\,B(X,Y)\Bigr),
\end{equation}
while~\eqref{eq:gauss_identity_ruled_surface} simplifies to
\begin{equation}
  K_\Sigma
  = \frac{
      g\bigl(R(X,Y)Y,X\bigr)
      - g\bigl(B(X,Y),B(X,Y)\bigr)
    }{\Delta}.
\end{equation}
\end{proposition}

\begin{proof}
Item (i) is the inverse of the $2\times2$ matrix
$\bigl(g_\Sigma(\partial_a,\partial_b)\bigr)=\bigl(\begin{smallmatrix}E&F\\F&G\end{smallmatrix}\bigr)$.
Item (ii) follows from
$\vec H=\tfrac12 g_\Sigma^{ab}B(\partial_a,\partial_b)$ and symmetry
of $B$. Item (iii) is the Gauss equation for a two-dimensional timelike
submanifold, evaluated on $(X,Y)$ and divided by
$\Delta=g(X,X)g(Y,Y)-g(X,Y)^2$. If the rulings are ambient geodesics,
then $\nabla_X X=0$, hence $B(X,X)=0$, which gives the stated
simplifications.
\end{proof}

For numerical surfaces constructed in
Section~\ref{sec:schwarzschild_example}, Proposition~\ref{prop:surface_metric_curvature_identities}
provides direct diagnostics: $\vec H$ measures deviation from minimal
world-sheets, while $K_\Sigma$ separates ambient-curvature contributions
from extrinsic bending through the $B$-terms.

\subsection{Jacobi fields along rulings and variation within the surface}
\label{subsec:jacobi_along_rulings}

We recall that in a (pseudo-)Riemannian manifold $(M,g)$, a smooth
one-parameter family of geodesics gives rise to a Jacobi field along
each geodesic, obtained by differentiating the family with respect to
the parameter. In the context of brachistochrone-ruled timelike
surfaces, the ``geodesic family'' is precisely the family of rulings,
so the corresponding Jacobi fields encode how the surface varies
transversely to each ruling.

\begin{figure}[t]
  \centering
  \begin{tikzpicture}[scale=0.8]
    \draw[->] (-0.2,0) -- (4.2,0) node[below] {$x=r\cos\varphi$};
    \draw[->] (0,-0.2) -- (0,4.2) node[left] {$y=r\sin\varphi$};

    \def\rzero{1.5}
    \def\rone{3.0}

    \draw[thick,gray!60] (0,0) circle (\rzero);
    \draw[thick,gray!60] (0,0) circle (\rone);
    \node[gray!60] at (\rzero+0.2,0.4) {$r_0$};
    \node[gray!60] at (\rone+0.2,0.4) {$r_1$};

    \draw[thick] 
      plot[domain=30:150,samples=100]
        ({\rzero*cos(\x)}, {\rzero*sin(\x)});
    \node at (0, \rzero+0.4) {$\alpha_0(s)$};

    \draw[thick]
      plot[domain=40:160,samples=100]
        ({\rone*cos(\x)}, {\rone*sin(\x)});
    \node at (0.3, \rone+0.4) {$\alpha_1(s)$};

    \foreach \shift in {0,1,2,-1,-2} {
      \draw[thick,blue!60]
        plot[domain=0:1,samples=50]
        ({ (1-\x)*\rzero*cos(70+\shift*8) + \x*\rone*cos(120+\shift*8)
           + 0.2*\x*(1-\x) },
         { (1-\x)*\rzero*sin(70+\shift*8) + \x*\rone*sin(120+\shift*8) });
    }

    \node[blue!60] at (-1.1,1.3) {\small $h^{\mathrm{J}}$-geodesics};
  \end{tikzpicture}
  \caption{Equatorial slice of the Schwarzschild exterior with two boundary
  curves $\alpha_0(s)$ and $\alpha_1(s)$ at radii $r_0$ and $r_1$. The
  time-minimizing timelike geodesics at fixed energy project to geodesics of
  the Jacobi metric $h^{\mathrm{J}}$, here shown schematically as a family of
  curved arcs connecting the two boundaries. Their lifts generate a
  brachistochrone-ruled timelike surface in spacetime.}
  \label{fig:schwarzschild_jacobi_schematic}
\end{figure}

\begin{definition}
\label{def:geodesic_variation}
Let $(M,g)$ be a Lorentzian (or more generally pseudo-Riemannian)
manifold. A smooth map
\[
  \Gamma : (-\varepsilon,\varepsilon)\times I \longrightarrow M,
\]
where $I\subset\mathbb{R}$ is an interval, is called a
\emph{geodesic variation} if:
\begin{enumerate}[label=(\roman*)]
  \item For each fixed $s\in(-\varepsilon,\varepsilon)$, the curve
        \(\gamma_s : I\to M\), \(\gamma_s(u) := \Gamma(s,u)\),
        is a geodesic (with respect to $g$).
  \item The map $\Gamma$ is smooth and its partial derivatives
        $\partial_s\Gamma$ and $\partial_u\Gamma$ are linearly
        independent wherever needed.
\end{enumerate}
For a fixed $s_0\in(-\varepsilon,\varepsilon)$, the vector field
$J:I\to TM$ along the geodesic $\gamma_{s_0}$ defined by
\begin{equation}
  J(u) := \partial_s \Gamma(s,u)\big|_{s=s_0}
\end{equation}
is called the \emph{variation field} or \emph{Jacobi field} associated
with the variation $\Gamma$ at $s_0$.
\end{definition}

It is well known that $J$ satisfies the \emph{Jacobi equation}
\begin{equation}
  \nabla_{\dot\gamma}\nabla_{\dot\gamma} J
  + R\bigl(J,\dot\gamma\bigr)\dot\gamma = 0,
  \label{eq:jacobi_equation}
\end{equation}
where $\gamma=\gamma_{s_0}$, $\dot\gamma=\partial_u\Gamma(s_0,u)$,
$\nabla$ is the Levi--Civita connection of $g$, and $R$ is the
Riemann curvature tensor.

We now specialize this general notion to the case of a
brachistochrone-ruled timelike surface
\(\Sigma:U\subset\mathbb{R}^2\to M\) as in
Definition~\ref{def:relativistic_brachistochrone_ruled}. To avoid
reducing the discussion to a tautological geodesic-variation statement,
we work directly with the reduced spatial variational problem from
Section~\ref{sec:relativistic_brachistochrone_ruled}.

\begin{assumption}
\label{ass:affine_param_rulings}
Let $\Sigma:U\to M$ be a relativistic
brachistochrone-ruled timelike surface as in
Definition~\ref{def:relativistic_brachistochrone_ruled}, where
$U\subset\mathbb{R}^2$ is an open set with coordinates $(s,u)$ and
let $\pi:M\to N$ be the stationary projection. For each fixed $s$, set
\[
  \gamma_s(u):=\Sigma(s,u),\qquad
  \sigma_s(u):=(\pi\circ\Sigma)(s,u).
\]
Assume:
\begin{enumerate}[label=(\roman*)]
  \item $\sigma_s$ is a $C^2$ extremal of the reduced functional
        \[
          T[\sigma]=\int_0^1 L(\sigma(u),\dot\sigma(u))\,\mathrm du
        \]
        between the moving endpoints
        $\alpha_0(s):=\sigma_s(0)$ and $\alpha_1(s):=\sigma_s(1)$.
  \item The map $(s,u)\mapsto\sigma_s(u)$ is smooth.
  \item (Strong Legendre condition, explicit form)
        For each fixed $s$, define
        \begin{equation}
          A_s(u):=\bigl(A_{ij}(s,u)\bigr)_{i,j=1}^n,
          \qquad
          A_{ij}(s,u):=\frac{\partial^2 L}{\partial v^i\partial v^j}
          \bigl(\sigma_s(u),\dot\sigma_s(u)\bigr).
        \end{equation}
        Then $A_s(u)=A_s(u)^{\mathsf T}$ and there exists
        $\lambda_s>0$ such that
        \begin{equation}
          \xi^{\mathsf T}A_s(u)\,\xi
          \ge \lambda_s\,\|\xi\|^2,
          \qquad
          \forall u\in[0,1],\ \forall \xi\in\mathbb{R}^n.
        \end{equation}
        In particular, $A_s(u)$ is invertible for all $u$.
\end{enumerate}
\end{assumption}

Under this assumption, the geometry of transverse variations is governed
by the linearization of the Euler--Lagrange equations of $L$.

\begin{proposition}
\label{prop:variation_is_jacobi}
Let $(M,g)$ be a Lorentzian manifold and
$\Sigma:U\to M$ a relativistic brachistochrone-ruled timelike
surface satisfying Assumption~\ref{ass:affine_param_rulings}.
Fix $s_0$ and consider the spatial ruling
\[
  \sigma_{s_0}(u):=(\pi\circ\Sigma)(s_0,u).
\]
Define the variation field
\begin{equation}
  W(u):=\partial_s\sigma(s,u)\big|_{s=s_0}
\end{equation}
along $\sigma_{s_0}$, viewed in local coordinates as a column vector
$W=(W^1,\dots,W^n)^{\mathsf T}$. Define, along
$(\sigma_{s_0}(u),\dot\sigma_{s_0}(u))$, the $n\times n$ matrices
\begin{equation}
  A(u):=\bigl(L_{v^iv^j}\bigr),\qquad
  B(u):=\bigl(L_{v^ix^j}\bigr),\qquad
  C(u):=\bigl(L_{x^ix^j}\bigr),
\end{equation}
where all derivatives of $L$ are evaluated at
$\bigl(\sigma_{s_0}(u),\dot\sigma_{s_0}(u)\bigr)$. Then $W$ satisfies
the linearized Euler--Lagrange equation in compact matrix form
\begin{equation}
  \mathcal{J}_L[W]
  :=\frac{\mathrm d}{\mathrm du}\Big(A(u)\dot W + B(u)W\Big)
   -B(u)^{\mathsf T}\dot W - C(u)W =0,
  \label{eq:l_jacobi_equation}
\end{equation}
or equivalently
\begin{equation}
  A\ddot W + (\dot A + B - B^{\mathsf T})\dot W + (\dot B - C)W=0,
\end{equation}
where $\dot A=\mathrm dA/\mathrm du$ and $\dot B=\mathrm dB/\mathrm du$.
In particular, in the static fixed-energy Jacobi case this reduces to
the classical Jacobi equation of the Riemannian metric $h^{\mathrm J}$
on $N$; and if the lifted rulings are additionally affinely
parametrized timelike geodesics of $(M,g)$, then
$J=\partial_s\Sigma|_{s=s_0}$ satisfies the ambient Jacobi equation
\eqref{eq:jacobi_equation}.
\end{proposition}

\begin{proof}
For each fixed $s$, $\sigma_s$ is an extremal of $T$, hence
\begin{equation}
  \frac{\mathrm d}{\mathrm du}\big(\nabla_v L(\sigma_s,\dot\sigma_s)\big)
  -\nabla_x L(\sigma_s,\dot\sigma_s)=0.
\end{equation}
Differentiate with respect to $s$ at $s=s_0$ and use
$W=\partial_s\sigma|_{s=s_0}$.
By the chain rule,
\begin{align}
  0
  &=\frac{\mathrm d}{\mathrm du}
    \Big(L_{vv}\,\dot W+L_{vx}\,W\Big)
    -\Big(L_{xv}\,\dot W+L_{xx}\,W\Big)\nonumber\\
  &=\frac{\mathrm d}{\mathrm du}\Big(A\dot W+BW\Big)
    -\Big(B^{\mathsf T}\dot W+CW\Big),
\end{align}
which is exactly \eqref{eq:l_jacobi_equation}. The last statement follows
from the standard identification of this linearization with the Jacobi
equation in the metric special cases described above.
\end{proof}

The boundary conditions for the linearized field $W$ along each
spatial ruling are determined by the variation of the projected
boundary curves.

\begin{lemma}
\label{lem:jacobi_boundary_values}
Let $\Sigma:U\to M$ be as in
Proposition~\ref{prop:variation_is_jacobi}, and assume that the
parameter domain $U$ contains a rectangle of the form
$I_s\times[0,1]$, with $I_s\subset\mathbb{R}$ an interval, such that
the projected map
\[
  \sigma(s,u):=(\pi\circ\Sigma)(s,u)
\]
has smooth boundary curves
\begin{equation}
  \alpha_0(s) := \sigma(s,0),\qquad
  \alpha_1(s) := \sigma(s,1),
  \qquad s\in I_s,
\end{equation}
are smooth. Fix $s_0\in I_s$ and let
\[
  \sigma_{s_0}(u) := \sigma(s_0,u),\qquad u\in[0,1],
\]
be the corresponding spatial ruling, with linearized field
$W(u)=\partial_s\sigma(s,u)|_{s=s_0}$ as in
Proposition~\ref{prop:variation_is_jacobi}. Then the boundary values
of $W$ are given by
\begin{equation}
  W(0) = \alpha_0'(s_0),\qquad
  W(1) = \alpha_1'(s_0),
\end{equation}
where the prime denotes derivative with respect to $s$.
\end{lemma}

\begin{proof}
By definition,
\[
  W(0)
  = \partial_s\sigma(s,u)\big|_{s=s_0,u=0}
  = \frac{\mathrm{d}}{\mathrm{d}s}\sigma(s,0)\big|_{s=s_0}
  = \alpha_0'(s_0),
\]
and similarly
\[
  W(1)
  = \partial_s\sigma(s,u)\big|_{s=s_0,u=1}
  = \frac{\mathrm{d}}{\mathrm{d}s}\sigma(s,1)\big|_{s=s_0}
  = \alpha_1'(s_0).
\]
\end{proof}

Thus, along each spatial ruling, the linearized field $W$ is uniquely
determined by the derivatives of the boundary curves at the endpoints
and by equation~\eqref{eq:l_jacobi_equation}. In the metric special
cases where this linearization identifies with a geometric Jacobi
equation, the same boundary-data principle applies to the corresponding
Jacobi fields. In particular, if the family of rulings arises from a
family of time-minimizing geodesics for a given time functional, then
the loss of global minimality along a ruling (e.g.\ at a cut point)
is reflected in the behavior of this linearized field, for instance
through the occurrence of conjugate points. (see
Figure~\ref{fig:schwarzschild_jacobi_geodesics}).

\begin{figure}
    \centering
    \includegraphics[width=0.9\linewidth]{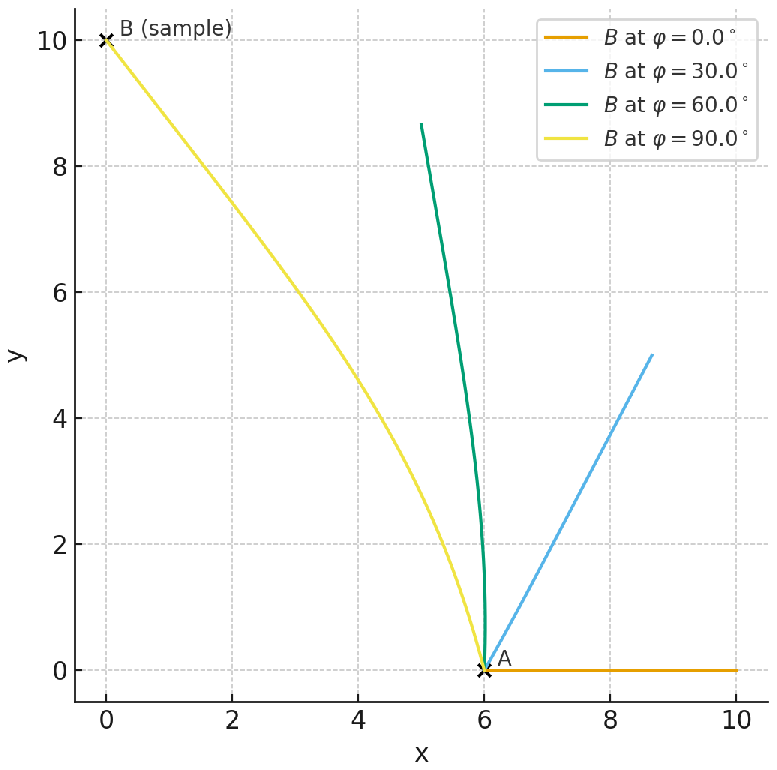}
    \caption{Family of Jacobi geodesics in Schwarzschild exterior}
    \label{fig:schwarzschild_jacobi_geodesics}
\end{figure}

\begin{remark}
\label{rem:normal_tangential_J}
Let $\Sigma:U\to M$ be a timelike immersed surface and let
$J=\partial_s\Sigma$ be the Jacobi field along a ruling
$\gamma_{s_0}$ as above. At each point of the ruling the tangent space
of the surface is spanned by $\partial_u\Sigma$ and $\partial_s\Sigma$,
so one may decompose $J$ into tangential and normal components
\[
  J = J^{\top} + J^{\perp},
\]
where $J^{\top}$ is tangent to the surface and $J^{\perp}$ is normal.
The tangential part $J^{\top}$ corresponds to a reparametrization of the ruling within the surface, whereas the normal part $J^{\perp}$ measures the transverse deviation of neighboring rulings in ambient spacetime. The normal component $J^{\perp}$ satisfies a scalar Jacobi-type equation that involve the ambient curvature and the second fundamental form of the surface; this provides a natural link between the stability of time-minimizing rulings and the extrinsic geometry of the brachistochrone-ruled surface.
\end{remark}

\section{Conclusions}
\label{sec:conclusion}

In this work, we developed \emph{brachistochrone-ruled timelike surfaces} as a geometric framework for families of arrival-time-minimizing trajectories. Conceptually, the key shift is from isolated two-point optimization problems to a surface-level description of time-optimal propagation between moving endpoint families.

The main outcomes are:
\begin{enumerate}
  \item A precise Newtonian and relativistic definition of brachistochrone-ruled surfaces (Definitions~\ref{def:brachistochrone_ruled_newtonian} and \ref{def:relativistic_brachistochrone_ruled}).
  \item A reduced stationary-spacetime formulation in which arrival-time minimization becomes a spatial first-order variational problem, with Jacobi-metric reduction in static fixed-energy regimes.
  \item Concrete realizations across model spacetimes: a cycloidal Newtonian toy model, a Minkowski consistency check with straight timelike rulings, and a Schwarzschild construction pipeline based on Jacobi geodesics.
  \item A first geometric variation analysis through the linearized ruling equation \eqref{eq:l_jacobi_equation}, together with boundary-data control (Proposition~\ref{prop:variation_is_jacobi} and Lemma~\ref{lem:jacobi_boundary_values}).
\end{enumerate}

Taken together, these results clarify the relation between brachistochrone theory and ruled-surface geometry and provide a practical bridge between abstract variational principles and computable spacetime examples.

A natural next step is the Schwarzschild--de~Sitter (Kottler) case \citep{kottler1918,gibbons1977}, where the static region is bounded by black-hole and cosmological horizons and the Jacobi reduction remains available with
\(
\beta(r)=1-2M/r-\Lambda r^2/3
\).
This setting should sharpen questions about static radius effects, caustics, and cut loci for time-minimizing rulings \citep{stuchlik2004,kagramanova2008}.

More broadly, promising directions include: (i) extrinsic curvature and stability theory for these surfaces; (ii) rotating stationary backgrounds such as Kerr; and (iii) null and genuinely non-stationary extensions, where the interaction between causality and variational structure is expected to be richer \citep{Perlick00}.

In short, brachistochrone-ruled timelike surfaces offer a compact language for time-optimal transport in curved spacetime, with clear geometric content and direct paths toward analytical and numerical generalizations.

\section{Acknowledgements}
The author thanks undergraduate student Ibrahim H.I. Abushawish for his valuable assistance in providing animated visualizations (\href{https://github.com/ibeuler/Brachistochrone-ruled_timelike_surfaces.git}{GitHub: Brachistochrone-ruled timelike surfaces}) that facilitated a better understanding of the results presented in this paper, and for his discussions that helped shape the perspective on possible future research directions.

\section*{Conflict of Interest Statement}
The author declares that there is no conflict of interest.

\section*{Funding}
No funding was received for this work.

\section*{Ethics Statement}
This theoretical study does not involve experiments on humans or animals.

\section*{Data Availability Statement}
No new experimental data were generated. Data supporting the findings of this study, including numerical outputs used for the figures, are available from the corresponding author upon reasonable request.

\bibliographystyle{unsrtnat}

\begin{thebibliography}{99}

\bibitem[Bao et al.(2000)]{BaoChernShen00}
Bao, D., Chern, S.-S., \& Shen, Z. 2000, \emph{An Introduction to Riemann--Finsler Geometry}
(New York: Springer)

\bibitem[Bao et al.(2004)]{BaoRoblesShen04}
Bao, D., Robles, C., \& Shen, Z. 2004, J. Diff. Geom., 66, 377

\bibitem[Bernoulli(1696)]{Bernoulli1696} 
Bernoulli, J. 1696, Acta Eruditorum, 18, 269

\bibitem[Caponio et al.(2011{\natexlab{a}})]{CaponioJavaloyesMasiello11}
Caponio, E., Javaloyes, M. Á., \& Masiello, A. 2011, Math. Ann., 351, 365

\bibitem[Caponio et al.(2011{\natexlab{b}})]{CaponioJavaloyesSanchez11}
Caponio, E., Javaloyes, M. Á., \& Sánchez, M. 2011, Rev. Mat. Iberoam., 27, 919

\bibitem[Caponio et al.(2024)]{CaponioJavaloyesSanchez24}
Caponio, E., Javaloyes, M. Á., \& Sánchez, M. 2024, Mem. Am. Math. Soc., 300, no.~1501

\bibitem[Erlichson(1999)]{Erlichson99}
Erlichson, H. 1999, Eur. J. Phys., 20, 299

\bibitem[Fortunato et al.(1995)]{FortunatoGiannoniMasiello95}
Fortunato, D., Giannoni, F., \& Masiello, A. 1995, J. Geom. Phys., 15, 159

\bibitem[Giannoni \& Piccione(1998)]{GiannoniPiccione98}
Giannoni, F., \& Piccione, P. 1998, J. Math. Phys., 39, 6137

\bibitem[Gibbons \& Hawking(1977)]{gibbons1977}
Gibbons, G. W., \& Hawking, S. W. 1977, Phys. Rev. D, 15, 2738

\bibitem[Gibbons et al.(2009)]{GibbonsHerdeiroWarnickWerner09}
Gibbons, G. W., Herdeiro, C. A. R., Warnick, C. M., \& Werner, M. C. 2009,
Phys. Rev. D, 79, 044022

\bibitem[Gibbons(2016)]{Gibbons16}
Gibbons, G. W. 2016, Class. Quantum Grav., 33, 025004

\bibitem[Goldstein et al.(1986)]{GoldsteinBender86}
Goldstein, H. F., \& Bender, C. M. 1986, J. Math. Phys., 27 (2)

\bibitem[Kagramanova et al.(2008)]{kagramanova2008}
Kagramanova, V., Kunz, J., \& Lämmerzahl, C. 2008, Gen. Relativ. Gravit., 40, 1275

\bibitem[Kottler(1918)]{kottler1918}
Kottler, F. 1918, Ann. Phys., 361, 401

\bibitem[Perlick(1990)]{Perlick90}
Perlick, V. 1990, Class. Quantum Grav., 7, 1849

\bibitem[Perlick(1991)]{Perlick91}
Perlick, V. 1991, J. Math. Phys., 32, 3148

\bibitem[Perlick(2000)]{Perlick00}
Perlick, V. 2000, \emph{Ray Optics, Fermat's Principle, and Applications to General Relativity}(Berlin: Springer)

\bibitem[Randers(1941)]{Randers41}
Randers, G. 1941, Phys. Rev., 59, 195

\bibitem[Stuchlík \& Slaný(2004)]{stuchlik2004}
Stuchlík, Z., \& Slaný, P. 2004, Phys. Rev. D, 69, 064001

\bibitem[Zermelo(1931)]{Zermelo31}
Zermelo, E. 1931, Z. Angew. Math. Mech., 11, 114

\end{thebibliography}

\end{document}